\documentclass{article}
\pdfoutput=1
\usepackage{amsthm,amsmath,natbib,amsfonts,graphicx}
\usepackage[affil-it]{authblk}
\RequirePackage[colorlinks,citecolor=blue,urlcolor=blue]{hyperref}

\theoremstyle{remark}
\newtheorem{remark}{Remark}
\theoremstyle{thm}
\newtheorem{thm}{Theorem}
\newcommand{\R}{{\mathbb{R}}}
\begin{document}

  \title{Adaptive gPCA: A method for structured dimensionality
    reduction}

\author{Julia Fukuyama}

\affil{Department of Statistics, Stanford University}

\maketitle

\begin{abstract}
  When working with large biological data sets, exploratory analysis
  is an important first step for understanding the latent structure
  and for generating hypotheses to be tested in subsequent
  analyses. However, when the number of variables is large compared to
  the number of samples, standard methods such as principal components
  analysis give results which are unstable and difficult to
  interpret. 

  To mitigate these problems, we have developed a method which allows
  the analyst to incorporate side information about the relationships
  between the variables in a way that encourages similar variables to
  have similar loadings on the principal axes. This leads to a
  low-dimensional representation of the samples which both describes
  the latent structure and which has axes which are interpretable in
  terms of groups of closely related variables.
  
  The method is derived by putting a prior encoding the relationships
  between the variables on the data and following through the analysis
  on the posterior distributions of the samples. We show that our
  method does well at reconstructing true latent structure in
  simulated data and we also demonstrate the method on a dataset
  investigating the effects of antibiotics on the composition of
  bacteria in the human gut.
\end{abstract}

\section{Introduction}

When analyzing biological data, we are often presented with a large
data matrix of interest along with side information about the
relationships between the variables in the data set. For example, in
microbiome data analysis, we have a data matrix containing abundances
of bacterial species as well as information about the phylogenetic
relationships between the bacteria. When analyzing transcriptome data,
we might have a data matrix with gene expression levels in the various
samples as well as information about which pathways the genes are
involved in. In light of this, many methods have been developed to
perform statistical analyses while taking into account the structure
of the variables. The fused lasso and its variations are often applied
to genomic data (\cite{tibshirani2008spatial, tibshirani2005sparsity,
  rinaldo2009properties}).  Kernel-penalized regression
(\cite{randolph2015kernel}) was developed to incorporate phylogenetic
structure into regression for microbiome data. The structure encoded
by gene networks has also been used to aid in classification of
microarray data (\cite{rapaport2007classification}), and regression
analysis of genomic data (\cite{li2008network}).

The current paper presents a new method for exploratory analysis of
such data which incorporates information about the relationships
between the variables. As motivation for why we might want to include
information about the relationships between the variables, consider
doing PCA on just the data matrix: we know that PCA is inconsistent
when the number of variables is much higher than the number of samples
(\cite{johnstone2012consistency}), which it usually is in modern
datasets. A common solution to the problem of inconsistency is to
assume that the principal axes are sparse and estimate them using a
regularized version of PCA which encourages sparse principal
axes. However, there are situations in which we do not expect
sparsity, but do expect other sorts of structure in the data. Our
method is designed to perform regularization in these situations.

From a more practical point of view, PCA is undesirable as an
exploratory method in situations where we have a large number of
variables because the variable loadings on the principal axes are
difficult to interpret: first of all, each axis is a linear
combination of all the variables, and moreover, the loadings of the
variables will not be structured according to our prior knowledge
about the relationships between them. Our method ensures that
variables which are similar to each other have similar axis
loadings. This leads to a more parsimonious explanation of the axes in
terms of groups of related variables, which should be more
interpretable and biologically relevant. In this regard it has a
similar aim as PCA with a fused lasso penalty on the variable loadings
(\cite{witten2009penalized}), although it works with more general
structures on the variables.

The layout of the paper is as follows: We first introduce a motivating
example in which incorporating outside information about the variables
is particularly important and which we will later use to illustrate
our method. We review generalized PCA, and then show how including a
prior in conjunction with the appropriate generalized PCA leads to our
new method, adaptive gPCA. To get a better understanding of adaptive
gPCA, we show how it is related to existing methods and we demonstrate
its performance on simulated and real data.

\section{Motivating example}\label{Sec:motivation}

The motivation for this work was our experience analyzing microbiome
data. In this paper, we will focus on one particular microbiome data
set, first published in \cite{dethlefsen2011incomplete}. The goal of
the study was to understand the effect of antibiotics on the
composition of bacteria in the human gut.  To this end, stool samples
were collected from each of three individuals before, during, and
after administration of two courses of the antibiotic
Ciprofloxacin. Between 52 and 56 samples were collected from each
individual for a total of 162 samples.

To understand what kinds of bacteria were present and at what
abundances, a certain highly variable segment of the 16S rRNA gene was
amplified by PCR and sequenced using next-generation sequencing. The
sequence of the variable segment of this gene was used as a proxy for
species. The species defined in this way are known in the microbiome
literature as operational taxonomic units (OTUs) and not species since
there is not necessarily a direct correspondence between them and
previously identified bacterial species. In the original analysis of
this data, OTUs were defined by clustering together sequences with at
least 95\% sequence identity using the Uclust software
(\cite{edgar2010search}), and the abundance of each OTU was defined as
the number of sequences mapping to the cluster. Clustering sequences
with at least 95\% sequence identity gave rise to a total of 2582
OTUs.

After defining OTUs in this way, the consensus sequence for each OTU
was mapped to a reference phylogenetic tree from the Silva 100
reference database (\cite{quast2013silva}). This mapping provides us
with the phylogenetic relationships between the bacteria corresponding
to the sequences that were obtained from the samples.

\subsection{The bacterial species problem}



No matter how we define OTUs, there is an underlying biological issue
in the definition of a bacterial species. Even today, there is a
division among microbiologists about whether bacterial species reflect
real underlying biology or whether they are primarily for scientists'
convenience. On the pro-species side, the ``ecotype'' theory described
in \cite{cohan2002bacterial} gives theoretical justification for why
we would expect to see groups of bacteria with much smaller within-
than between-group sequence divergence and why these are meaningful
biological units. The anti-species side of the debate cites as
evidence the large amount of lateral gene transfer and homologous
recombination as well as the amount of genetic dissimilarity within
groups traditionally defined as species. For an example of this type
of argument, see \cite{doolittle2006genomics}.

However, following Darwin who wrote that ``all true classification is
genealogical,'' most microbiologists agree on the usefulness of the
phylogenetic tree for describing the relationships between
bacteria. For example, \cite{brenner2005classification},
\cite{doolittle2006genomics}, \cite{cohan2002bacterial}, all agree on
this despite differing on the existence of bacterial
species. Therefore, to bring our statistical methods more in line with
biological understanding, methods that deal with bacterial species
should incorporate the phylogeny instead of implicitly assuming that
species are all equally distinct.


\subsection{Existing methods for incorporating phylogeny in microbiome
  data analysis}

Several methods have been proposed for including phylogenetic
information in exploratory data analysis. Some examples are double
principal coordinates analysis (DPCoA), which was originally described
in \cite{pavoine2004dissimilarities} as a method for incorporating
more general structure about the variables but which can accommodate
phylogenetic structure, weighted and unweighted Unifrac
(\cite{lozupone2005unifrac} and \cite{lozupone2007quantitative}) which
were developed specifically for microbiome data, a number of variants
of the Unifrac distances including generalized Unifrac
(\cite{chen2012associating}) and variance-adjusted weighted Unifrac
(\cite{chang2011variance}), and edge PCA
(\cite{matsen2013edge}). Unfortunately, many of these methods tend to
implicitly group together species at a very high taxonomic level,
which is not always desirable. Although a high-level grouping might
lead to good insights in some situations, in general we would like a
more flexible method where we can tune how coarse or fine of an
analysis to perform.

Another issue with many of the existing methods for incorporating the
phylogeny (in particular Unifrac and its variants) is that they are
distance-based and when they are applied in conjunction with
multi-dimensional scaling they give axes with no interpretation in
terms of the species. Since we are interested in dimensionality
reduction for hypothesis generation and for understanding the biology
underlying the structure we see in the data, it is important for the
method to also give insight into which species are responsible for any
clustering or gradients we see in the low-dimensional representation
of the samples. In contrast to most of the existing methods, the
procedure we introduce in this paper will use the phylogenetic
relationships between the bacterial species to give interpretations of
the axes in terms of groups of closely related species, which we
expect to be more easily interpretable and to lead to a better
understanding of the differences between microbial communities.

\subsection{Other properties of the antibiotic dataset}
The data set from \cite{dethlefsen2011incomplete} that we are
considering in this paper also has many of the features we discussed
in the introduction. We have 2582 variables (the abundances of the
species or OTUs) and only 162 samples, making the variable loadings
from PCA difficult to interpret and unreliable. We also do not expect
sparsity in the principal axes. The main divisions in the data are
samples from different individuals and samples taken during
administration of the antibiotic vs. not, and we do not expect either
of these divisions to be associated with changes in only a few
species. On the contrary, we expect the administration of the
antibiotic to change the relative abundances of nearly all of the
species, and we know from other microbiome studies that different
individuals have very different gut microbiome compositions at the
species level (see \cite{shade2012beyond}). On the other hand, we do
expect phylogenetically similar species to react in similar ways to
the antibiotic. For all of these reasons, we expect a method which
incorporates the phylogeny to be useful in understanding these data.

\section{Generalized PCA}

Before we introduce adaptive gPCA, we first review generalized PCA
(gPCA) and give some intuition about the kinds of solutions it
produces. Generalized PCA has already been used to create structured
low-dimensional data representations: For the particular case of
analyzing microbiome data with a phylogenetic tree, it was shown that
double principal coordinates analysis
(\cite{pavoine2004dissimilarities}), which we will look at in more
detail later, could be re-expressed as a gPCA
(\cite{purdom2011analysis}). In a rather different context, but also
for the purpose of incorporating the structure of the variables into
the analysis, the method for functional principal components
introduced in \cite{silverman1996smoothed} also has an interpretation
as PCA with respect to a non-standard inner product, or a generalized
PCA.

We follow the notation from the French multivariate
tradition in considering gPCA on a triple $(X, Q, D)$, where
$X \in \R^{n \times p}$ is our data matrix of $n$ samples measured on
$p$ variables, and $Q$ and $D$ are positive definite matrices with
$Q \in \R^{p \times p}$ and $D \in \R^{n \times n}$ (see
\cite{holmes2008multivariate} for a more thorough explanation). The
sample scores for gPCA on the triple $(X, Q, D)$ are the solutions to
the optimization problem
\begin{align}
\max_{u_i \in \R^n} &\quad u_i^T D XQX^T D u_i , \quad i = 1,\ldots, k
  \label{Eq:gpca_scores}\\
\text{s.t. } & \quad u_i^T D u_i  =1 , \quad i = 1,\ldots, k \nonumber\\
& \quad  u_i^T D u_j = 0, \quad 1 \le i < j \le k \nonumber
\end{align}
Similarly, the principal axes for gPCA on the triple $(X, Q, D)$ are
given by
\begin{align}
\max_{v_i \in \R^p} &\quad v_i^T Q X^T D XQ v_i , \quad i = 1,\ldots,
                      k \label{Eq:gpca_axes}\\
\text{s.t. } & \quad v_i^T Q v_i  =1 , \quad i = 1,\ldots, k \nonumber\\
& \quad  v_i^T Q v_j = 0, \quad 1 \le i < j \le k \nonumber
\end{align}

We can think of gPCA either as PCA in a non-standard inner product
space or as PCA on observations corrupted with non-spherical
noise. Both ways are informative and we review both here.

\subsection{Non-spherical noise}

Recall, following \cite{allen2014generalized}, that PCA can be
formulated as a maximum likelihood problem. Suppose that our observed
data is $X \in \R^{n \times p}$, and our model is
\begin{align*}
X &= U\Lambda V^T + E \\
E_{ij} &\overset{\text{iid}}{\sim} N(0, \sigma^2)
\end{align*}
where $U \in \R^{n \times k}$ and $V \in \R^{p \times k}$ are
orthogonal, and $\Lambda$ is diagonal. Then if the row scores,
principal axes, and variances of PCA on $X$ are given by
$\hat U \in \R^{n \times k}$, $\hat V \in \R^{p \times k}$, and
$\hat \Lambda \in \R^{k \times k}$, respectively, then the maximum
likelihood estimate of $U\Lambda V^T$ is
$\hat U \hat \Lambda \hat V^T$.

The generalized PCA solution is obtained when the elements of the
noise matrix $E$ are not independent and identically distributed. If
we change our model to
\begin{align*}
  X &\sim \mathcal{MN}_{n \times p}(U\Lambda V^T ,D^{-1}, Q^{-1})
\end{align*}
and if the row scores, principal axes, and variances of gPCA on the
triple $(X, Q, D)$ are given by $\hat U$, $\hat V$, and
$\hat \Lambda$, then the maximum likelihood estimate of $U\Lambda V^T$
is $\hat U \hat \Lambda \hat V^T$ (\cite{allen2014generalized}). This
allows us to account for more complicated error structures: we can
have correlation on the rows, on the columns, or both. The error
structure is not fully general --- it still must be separable --- but
this formulation allows for some dependence in the noise.

In practice, the assumption of normality of the errors may not be even
approximately true if our data is highly skewed or discrete, both of
which hold in our motivating example for those bacterial species with
low expected counts. In this case, we need to apply some sort of
transformation to the raw data so as to bring it more in line with our
assumptions. The correct transformation to use will depend on the data
in question, but for microbiome count data two common choices are to
use a started log transformation or to use the variance-stabilizing
transformation from the package DESeq2 (see \cite{mcmurdie2014waste}
and \cite{callahan2016bioconductor} for examples and the motivation
for this transformation). For the data analyzed in this paper, we
transform the counts using a started log transformation and remove
some of the bacterial species with particularly large fractions of
zero counts.

\subsection{Non-standard inner product}\label{Sec:ip}

The other way of thinking of gPCA on the triple $(X, Q, D)$ is simply
as PCA in a non-standard inner product space. Note that $Q$ and $D$,
being positive definite matrices, define inner products on $\R^p$ and
$\R^n$ in the following way:
\begin{align*}
\langle x, y \rangle_Q = x^T Q y, \quad x,y \in \R^p \\
\langle x,y \rangle_D = x^T D y, \quad x, y \in \R^n
\end{align*}
From the form of the gPCA problem as shown in (\ref{Eq:gpca_scores})
and (\ref{Eq:gpca_axes}), we see that gPCA is simply standard PCA with
the standard inner product replaced with the $Q$- and $D$- inner
product for the rows and columns respectively. In particular, gPCA of
the triple $(X, I, I)$ is equivalent to standard PCA.

To give some intuition into the reasons for and effects of working in
a non-standard inner product space, consider linear discriminant
analysis (LDA). In LDA, we have a (centered) data matrix
$X \in \R^{n \times p}$, and the samples fall into a set of $g$
groups. Suppose that we have weights for each sample, which are stored
on the diagonal of a matrix $D\in \R^{n \times n}$. Let
$Y \in \R^{n \times g}$ be an indicator matrix assigning samples to
groups, let $A \in \R^{g \times p}$ be a matrix containing the group
means for each of the $p$ variables, let $\Delta_Y = Y^T D Y$ be a
matrix containing the group weights, and let the within-group
covariance matrix be $W = (X - YA)^T D (X - YA)$.

With this notation, LDA can be written as gPCA on the triple
$(A, W^{-1}, \Delta_Y)$. We know that in LDA we want to find a
projection that maximizes the ratio of the between-class and the
within-class variance. We can think of this as LDA favoring
projections in directions of small within-class covariance, or
projections along axes $v$ for which $v^T W^{-1} v$ is
large. Analogously, if we have a more general gPCA on the triple
$(X, Q, D)$, we can think of the effect of the inner product matrix
$Q$ as favoring projections along axes $v$ for which $v^T Q v$ is
large.

In LDA, our inner product on the rows comes from the data, but we can
also imagine having an inner product on the rows which is based on
prior knowledge about the data. In what follows, we will choose an
inner product on the rows for which directions where similar variables
have similar scores are favored over directions in which similar
variables have dissimilar scores.

\begin{remark}
  Note that neither the correlated errors nor the non-standard inner
  product interpretation of gPCA are entirely satisfactory for the
  problem we want to solve. In our motivating example, we expect there
  to be axes which are both smooth on the tree and for which the
  projections of the samples have a large variance. 

  From the non-standard inner product interpretation, we know that we
  can design an inner product on the rows which will pull out axes
  with these properties. However, there are many ways to construct
  such inner product matrices and the non-standard inner product
  interpretation gives us very little insight into which one to
  choose.

  The other interpretation, in which we assume correlated errors, is
  also not quite right since it assumes structure in the error when we
  want to encode information about the structure of the signal. 
\end{remark}

\section{Adaptive gPCA}\label{Sec:adaptive}

In this section, we describe our proposal for incorporating prior
information about the structure of the variables. The basic idea is as
follows: We include a prior in our model which encodes our intuition
that the variables which are similar to each other should behave in
similar ways (in the case of microbiome data the idea is that species
close together on the tree will behave similarly). We perform
generalized PCA on the posterior estimate of each sample given the
data, taking into account the variance structure of the
posterior. Varying the scalings of the prior and noise variances gives
rise to a one-dimensional family of generalized PCAs which favor
progressively smoother solutions according to the structure of the
variables. Our method, adaptive gPCA, chooses which member of the
family to use by estimating the scalings of the signal and the noise
by maximum marginal likelihood.

\subsection{Data model}\label{Sec:agpcadatamodel}

Suppose we have a positive definite similarity matrix
$Q \in \R^{p \times p}$ (a kernel matrix) between the variables.  To
prevent scaling issues, assume that $\text{tr}(Q) = p$. Note that
since $Q$ is positive definite, it is also a covariance matrix, and a
random vector with covariance $Q$ will have stronger positive
correlations between variables which are more similar to each
other. For microbiome data with a phylogenetic tree, we will take $Q$
to be the matrix where $Q_{ij}$ represents the amount of shared
ancestral branch length between species $i$ and $j$. We use this
kernel matrix for several reasons, one of which is that it is the one
implicitly used in DPCoA; it is also related to the covariance of a
Brownian motion run along the branches of the tree.

With this in mind, consider the following model for our data
matrix $X$:
\begin{align}
\mathbf x_i \overset{\text{iid}}{\sim} N(\mathbf \mu_i, \sigma_2^2 I), \quad i = 1,\ldots, n\label{Eq:model2}\\
\mathbf \mu_i \overset{\text{iid}}{\sim} N(0, \sigma_1^2 Q), \quad i = 1,\ldots, n\label{Eq:model1}
\end{align}
Here we are simply including a prior in our model. The prior
incorporates information about the structure in our variables: since
the $\mu_i$'s have covariance equal to a scalar multiple of $Q$,
inference using this prior will allow us to regularize towards this
structure, or to smooth the data towards our expectation that similar
variables will behave in similar ways.

\subsection{PCA on Bayes estimates}
We are interested in the ``true'' values given in $\mu_i$ and not the
observed data $\mathbf x_i$, and so the appropriate next step is to
compute the posterior distribution of the the $\mu_i$'s and then
perform PCA on these posteriors. We can compute the posterior
distribution $\mathbf \mu_i \mid \mathbf x_i$ using Bayes' rule, which
is
\begin{align}
\mathbf \mu_i \mid \mathbf x_i = x \sim  N(\sigma_2^{-2} Sx, S)
\end{align}
with
\begin{align}
S = (\sigma_1^{-2} Q^{-1} + \sigma_2^{-2}I)^{-1}
\end{align}
Now we want to perform PCA on the posterior estimates of the
$\mathbf \mu_i$'s. We need to take into account the fact that the
posterior distributions for each $\mu_i$ have non-spherical variance,
and so we need to use gPCA instead of standard PCA. The method we use
to compute the sample scores and principal axes is described in the
following theorem:

\begin{thm} \label{Thm:pcabayes} The row scores from gPCA on the
  posterior estimates $\mu_i \mid \mathbf x_i$ from the model
  described in Section \ref{Sec:agpcadatamodel} are the same, up to a
  scaling factor, to the row scores from gPCA on $(X, S, I_n)$. The
  principal axes from gPCA on the posterior estimates are the same, up
  to a scaling factor, as the principal axes from gPCA on $(X, S,
  I_n)$ pre-multiplied by $S$. 
\end{thm}
\begin{proof}
See appendix. 
\end{proof}

From this theorem, we see that when we perform gPCA on the posterior
estimates obtained from the model described in Section
\ref{Sec:agpcadatamodel}, different scalings of the prior and the
noise variances simply lead to gPCAs with different row inner product
matrices.

\subsection{A family of gPCAs}

Now we can explore the family of inner product matrices which our
model gives rise to. Up to a scaling factor, the matrix
$S = (\sigma_1^{-2} Q^{-1} + \sigma_2^{-2}I)^{-1}$ depends only on the
relative sizes of $\sigma_1$ and $\sigma_2$, the scalings for the
prior and the noise. We therefore have a one-dimensional family of
gPCAs determined by the relative sizes of $\sigma_1$ and
$\sigma_2$. To get some insight into this family, we can first
consider the endpoints.

As $\sigma_1 / \sigma_2 \to 0$, that is, as the noise becomes very
small compared to the prior structure, $S$ becomes more and more like
a scalar multiple of the identity, and so we approach a scalar
multiple of gPCA on the triple $(X, I, I)$, or standard PCA. At the
other end, as $\sigma_2 / \sigma_1 \to 0$, we approach a scalar
multiple of gPCA on the triple $(X, Q, I)$. The gPCA on $(X, Q, I)$
turns out to be very closely related to double principal coordinates
analysis (DPCoA, originally described in
\cite{pavoine2004dissimilarities}), which is another method for
incorporating information about the variables into the analysis. We
will describe DPCoA and its relationship with our method further in
Section \ref{Sec:relationshipsdpcoa}, but for now it suffices to note
that this family of gPCAs can be thought of as interpolating between
DPCoA and standard PCA or as giving us a tunable parameter controlling
how smooth we want the principal axes to be.

We might also wonder why this family is better than other families we
might consider. A possibly more natural method would be one where we
add a ridge penalty to $Q$, resulting in gPCA on
$(X, Q + \lambda I, I)$.  This family has the same endpoints as the
family we have described: when $\lambda = 0$ we have gPCA on
$(X, Q, I)$, and as $\lambda \to \infty$ we get standard PCA. The
difference between the two is the path between the two endpoints. Very
roughly, when we add a ridge penalty to $Q$, the main effect is to
increase the small eigenvalues, but when we add a ridge penalty to
$Q^{-1}$ we make the large eigenvalues more similar to each other. In
general, the small eigenvalues of $Q$ correspond to eigenvectors that
are very rough (the values are very different for variables which are
similar to each other), while the large eigenvalues correspond to
eigenvectors that are smooth. When we do structured dimensionality
reduction, we are almost always going to want to dampen any variance
along rough eigenvectors, but we don't necessarily prefer variance in
the direction of an extremely smooth eigenvector over variance in the
direction of a mostly-smooth eigenvector. When we use $Q + \lambda I$,
we remove the dampening on the rough directions, but when we use
$S = (\sigma_1^{1} Q^{-1} + \sigma_2^{-2} I)^{-1}$ we keep the
eigenvalues of the rough directions small and decrease the difference
between eigenvalues of smooth eigenvectors.

\subsection{Automatic selection of family member}

So far, we have been assuming that $\sigma_1$ and $\sigma_2$ are
known, but this is not generally going to be the case. It is possible
to choose values for the two purely subjectively, based on how heavily
you want to weight your prior knowledge about the variables compared
to the actual data. However, if choosing subjectively is not
appealing, the structure of the model suggests that we can estimate
the values $\sigma_1$ and $\sigma_2$ from the data itself by maximum
marginal likelihood. To be more concrete, according to our data model
we have
\begin{align}
\mathbf x_i \overset{\text{iid}}{\sim} N(0, \sigma_1^2 Q + \sigma_2^2 I)
\end{align}
The overall log likelihood of the data is therefore (up to a constant factor)
\begin{align}
  \ell(X; \sigma_1, \sigma_2) = -\frac{n}{2}\log |\sigma_1^2 Q + \sigma_2^2
  I| - \sum_{i=1}^n\frac{1}{2} \mathbf x_i^T
  (\sigma_1^2 Q + \sigma_2^2 I)^{-1} \mathbf x_i \label{Eq:loglik}
\end{align}
Maximizing this likelihood is not a convex problem and there does not
appear to be a closed-form solution, but it is possible to transform
it into a problem of optimizing one parameter over the unit
interval. To do this, we introduce some new notation. Let $r =
\sigma_1^2 / (\sigma_1^2 + \sigma_2^2)$, and let $\sigma^2 =
\sigma_1^2 + \sigma_2^2$. Let $Q = V\Lambda V^T$ be the
eigendecomposition of $Q$ where $V$ is an orthogonal matrix and
$\Lambda$ is diagonal containing the eigenvalues $\lambda_1, \ldots,
\lambda_p$. Finally, let $\mathbf {\tilde x}_i = V^T \mathbf x_i$ and
$\tilde x_{ij}$ be the $j$th element of $\tilde {\mathbf x}_i$. The
log likelihood in the new parameterization is 
\begin{align}
\ell(X; r, \sigma) &= -\frac{np}{2} \sigma^{2} \log |r Q + (1 -r)I| -
  \sigma^{-2}\sum_{i=1}^n \frac{1}{2} \mathbf x_i^T (r Q + (1 -r)I)
  \mathbf x_i \\
&= -\frac{np}{2} \sigma^{2} \sum_{j=1}^p \log(r \lambda_j + 1-r) -
  \sigma^{-2} \sum_{i=1}^n \sum_{j=1}^p \frac{1}{2} \frac{\tilde
  x_{ij}^2}{r \lambda_j  + 1-r}
\end{align}
Based on the expression above, we can find a closed-form solution for
the maximizing value of $\sigma^2$ for any fixed $r$. This gives us
\begin{align}
{\sigma^{2}}^*(r) = \frac{1}{np}\sum_{i=1}^n \sum_{j=1}^p \tilde x_{ij}^2 / (r \lambda_i + 1
  -r)
\end{align}
We can then re-write the likelihood as a function of $r$ only. This is
still not convex and does not have a closed-form solution, but since
we now have only one parameter which lies on the unit interval, the
optimization can be performed numerically.

\begin{remark}
  We can get some insight into what sorts of solutions adaptive gPCA
  will choose by considering some extreme cases. First, consider the
  case where the covariance of $X$ is equal to $Q$. In this case, the
  automatic method will set the noise scaling $\sigma_2$ equal to 0,
  which corresponds to gPCA on $(X, I, I)$ or standard PCA. On the
  other hand, if the covariance of $X$ is spherical, the prior or
  signal scaling will be set equal to zero, corresponding to gPCA on
  $(X, Q, I)$. Therefore, when the marginal covariance is already
  structured according to the prior information on the variables, we
  don't do any regularization towards the prior structure. On the
  other hand, when it doesn't seem like the marginal covariance is
  structured according to the prior information on the variables, we
  do the maximum amount of regularization towards the prior
  structure. We can think of this as trying to balance the competing
  objectives of obtaining a gPCA plot which reflects the directions of
  maximum variation in the data and one which gives similar variables
  similar axis loadings. 
\end{remark}

\begin{remark}[Choice of $Q$]
  $Q$ can be any positive definite kernel matrix between the
  variables. A kernel matrix is often a natural way to encode
  relationships between variables: for example, if the variables are
  the nodes in a graph, there are many graph kernels available to
  describe the similarities between the nodes, mostly based on the
  graph Laplacian. For some examples, see \cite{kondor2002diffusion}.

  If we start off with Euclidean distances between variables instead
  of similarities, a natural way to create a kernel matrix is as
  follows: Suppose $\delta \in \R^{p \times p}$ is a matrix with the
  squared distances between the variables, and let
  $P = I - \mathbf 1_p \mathbf 1_p^T / p$ be the centering
  matrix. Then, if the distances implied by $\delta$ are Euclidean,
  $-P \delta P$ is a positive definite similarity matrix. This matrix
  contains the inner products between points if they are embedded in
  $\R^p$ such that the distances between them match the distances
  implied by $\delta$ and they are centered around the origin.
\end{remark}

\subsection{Adaptive gPCA} 
Putting everything together, we have the following method. We start
out with a data matrix $X \in \R^{n \times p}$ and either a kernel
matrix $Q \in \R^{p \times p}$ containing similarities between the
variables or a matrix $\delta \in \R^{p \times p}$ containing the
squared distances between the variables (we assume the set of
distances is Euclidean). We perform the following steps:
\begin{enumerate}
\item If we started with distances between the variables, set
  $Q = P (-\delta/2) P$. Since the distances are Euclidean, this
  definition of $Q$ gives a positive definite kernel matrix. Otherwise
  use the kernel matrix provided.
\item Find $\sigma_1$ and $\sigma_2$ which maximize the likelihood
  function in equation (\ref{Eq:loglik}) corresponding to the model in
  (\ref{Eq:model1})-(\ref{Eq:model2}). 
\item Let $S = (\sigma_1^{-2} Q^{-1} + \sigma_2^{-2} I)^{-1} $.
  Perform gPCA on the triple $(X, S, I)$.  The sample scores for
  adaptive gPCA are given by the row scores of this gPCA, and the
  variable scores for adaptive gPCA are given by the column scores of
  this gPCA pre-multiplied by $S$.
\end{enumerate}

To understand why this method encourages principal axes with variable
loadings which are similar for variables which are similar to each
other, recall the description in Section \ref{Sec:ip} of LDA as a gPCA
on $(A, W^{-1}, D)$ (where $A$ is a matrix of group means, $W$ is the
within-class covariance matrix, and $D$ is a diagonal weight
matrix). The interpretation here is that the discriminant vectors $v$
are encouraged to be in directions where the within-class covariance
is small, or $v^T W^{-1}v$ is large. Similarly, gPCA on $(X, S, I)$
will encourage principal axes $v$ for which $v^T S v$ is large. Since
$S$ has the same eigenvectors with the same ordering of eigenvalues as
$Q$, the similarity matrix for the variables, this is the same as
encouraging principal axes $v$ which have similar loadings for
variables which are similar to each other.

\section{Relationship with DPCoA}\label{Sec:relationshipsdpcoa}

The family of gPCAs given by our method can be thought of as bridging
the gap between standard PCA and another method for incorporating
information on the structure of the variables, double principal
coordinates analysis (DPCoA), originally described in
\cite{pavoine2004dissimilarities}. Briefly, DPCoA is a method for
giving a low-dimensional representation of ecological count data
(generally the abundance of species at several sampling sites) taking
into account information about the similarities between species. DPCoA
starts with a matrix of Euclidean distances between the species and
the counts of each species at each sampling site. To obtain the DPCoA
ordination, we perform the following steps:
\begin{enumerate}
\item Perform a full multi-dimensional scaling on the species.
\item Place each sampling site at the center of mass of the species
  vector corresponding to that site. 
\item Perform PCA on the matrix of sampling site coordinates, and
  project both the sampling site points and the species points onto
  the PCA axes. 
\end{enumerate}
DPCoA was later shown to be equivalent to gPCA using a certain
non-standard inner product in \cite{purdom2011analysis} for the
special case of tree-structured variables, and it can be shown to be
equivalent to a gPCA given any Euclidean distance structure on the
variables.  The relationship between DPCoA with Euclidean distances
between the variables and gPCA is given in the following theorem.
\begin{thm} \label{Thm:dpcoagpca}
  Suppose we have a count matrix $X \in \R^{n \times p}$ and a set of
  Euclidean distances between the $p$ variables. We construct a matrix
  $\delta \in \R^{p \times p}$ containing the squares of the distances
  between the variables. Let  $w_L = X\mathbf 1 / \mathbf 1^T X \mathbf 1$,
  $w_S = X^T \mathbf 1 / \mathbf 1^T X \mathbf 1$, and for any weight
  vector $w$ let $P_w = I - \mathbf 1 w^T$ and $D_w$ denote the
  diagonal matrix with $w$ on the diagonal. Then:
\begin{enumerate}
\item The row scores from DPCoA on $X$ using the distances implied by
  $\delta$ are the same (up to a sign change) as the row scores
  obtained from gPCA on $(D_{w_L}^{-1}X P_{w_S}, P_{w_S}(-\delta / 2) P_{w_S}, D_{w_L})$.
\item If the column scores from gPCA on
  $(D_{w_L}^{-1}X P_{w_S},P_{w_S}(-\delta / 2) P_{w_S}, D_{w_L})$ are
  given by $Z$, then the column scores from DPCoA on $X$ using the
  distances implied by $\delta$ are the same (up to a sign change) as
  $P_{w_S}(-\delta / 2) P_{w_S}Z$.
\end{enumerate}
\end{thm}
\begin{proof}
See the appendix. 
\end{proof}

DPCoA was developed for count data, and in the French multivariate
tradition count data is typically analyzed by transforming the counts
into relative abundances and retaining the column and row sums as
weightings on the rows and columns (see, for example, the section on
correspondence analysis in \cite{holmes2008multivariate}). The row and
column sums need to be retained and used as weights since they give
the precision with which we know the relative abundance vectors for
each location. Therefore, in the gPCA formulation of DPCoA, we use
centering matrices which are weighted according to the variable
weights ($P_{w_S}$) and use an inner product on the columns which
weights the rows according to their counts ($D_{w_L}$). However, with
the more general kinds of data we are considering in this paper, we
will not necessarily have a measure of the precision with which the
variables are measured, and the natural adaptation of the method to
non-count data would be to weight all the variables equally. This
means setting $w_S = \mathbf 1 / p$ and $w_L = \mathbf 1/ n$. In this
case, the gPCA triple becomes $(XP, P(-\delta / 2)P, I)$ (with
$P = I - \mathbf 1 \mathbf 1^T /p$, a centering matrix). The inner
product matrix here is the limiting inner product matrix in our family
of gPCAs as $\sigma_2 / \sigma_1 \to 0$, and the data matrix is simply
a standard centered data matrix. Thus, we see that a small
modification of DPCoA adapting it to non-count data is equivalent to
one of the endpoints in our family of gPCAs.

\section{Simulation results}

To evaluate the performance of adaptive gPCA, we simulated data from
models in which we would hope for it to perform well. To match our
motivating example of microbiome abundance data with information about
the phylogenetic relationships between the bacteria, we suppose that
the variables are related to each other by a phylogenetic tree. We
used a random tree (using the function \texttt{rtree} in the
\texttt{ape} package \cite{paradis2004ape} in R) for the relationship
between the variables, and the similarity matrix
$Q \in \R^{p \times p}$ we use to encode the information about the
tree structure is defined as follows:
\begin{align}
Q = \mathbf 1 s^T + s \mathbf 1^T - \delta\label{Eq:Qdef}
\end{align}
where $s \in \R^p$ gives the distance between each leaf node and the
root and $\delta \in \R^{p \times p}$ gives the distance on the tree
between the leaf nodes. This definition gives us a matrix $Q$ with
$Q_{ij}$ proportional to the amount of shared ancestry between nodes
$i$ and $j$, and it is also equal to the covariance matrix of a
Brownian motion on the phylogenetic tree. For our two simulation
experiments, we will compare adaptive gPCA using $Q$ as the similarity
matrix to standard PCA and gPCA on $(X, Q, I)$, which is intended to
be a slight extension of DPCoA to real-valued data. 

\subsection{Simulation A}
For the first simulation, we generate our data matrix as rank-one plus
noise, and we ensure that the coefficients of the principal axis are
smooth on the tree. More specifically, we generate our data matrix
$X \in \R^{n \times p}$ as follows:
\begin{align}
X &= uv^T + E \\\label{Eq:X}
E_{ij} &\overset{\text{i.i.d.}}{\sim} N(0, \sigma^2), \quad i = 1,
         \ldots, n, \; j = 1, \ldots, p \\\label{Eq:E}
u_i & \overset{\text{i.i.d.}}{\sim} N(0,1), \quad i = 1,\ldots, n \\\label{Eq:u}
v & \sim N(\mathbf 0_p, V_{(m)}V_{(m)}^T)
\end{align}
where $V_{(m)} \in \R^{p \times m}$ denotes the matrix whose columns
are the top $m$ eigenvectors of $Q$. The value of $m$ governs how
smooth $v$ is: if $m$ is small, $v$ tends to have coefficients which
are very smooth and exhibit long-range positive dependence on the
tree, and as $m$ increases the coefficients get more and more
rough. At the extreme case of $m = p$, $V_{(m)}V_{(m)}^T = I_p$, and
so there is no relationship at all between the coefficients of $v$ and
the tree structure.

We compare adaptive gPCA to standard PCA and gPCA on $(X, Q, I)$
(intended to be similar to DPCoA), looking at the correlations between
the true and estimated scores and principal axes. We vary both $m$
(controlling the smoothness of the principal axis on the tree) and
$\sigma$ the error noise. The results are shown in Figure
\ref{Fig:SimA}. We see that both standard PCA and adaptive gPCA do a
perfect job at recovering both the principal axis and the scores when
there is no noise, while gPCA on $(X, Q, I)$ does poorly at recovering
the principal axis unless there is very strong long-range dependence
in the coefficients of the principal axis (corresponding to $m = 1$ in
the left-most column). The performance of all the methods degrades
with increasing noise, but the performance of adaptive gPCA falls off
less quickly than the performance of PCA when there is at least a
moderate amount of smoothness in the coefficients of the principal
axis.

\begin{figure}[h]
\includegraphics[width=\textwidth]{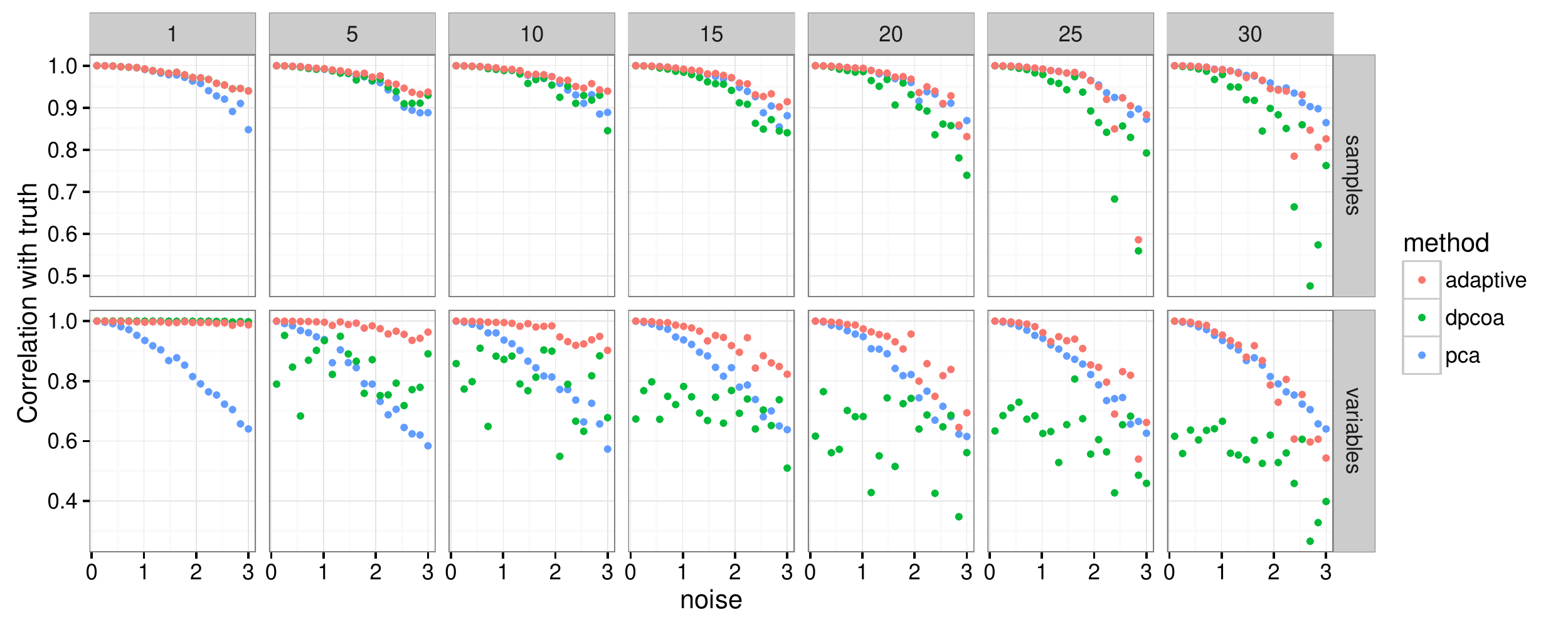}
\caption{Results from simulation A. Correlations between the true and
  estimated principal axis (top) and true and estimated scores
  (bottom) for different values of $m$ (columns, see text for
  explanation of $m$). }\label{Fig:SimA}
\end{figure}

\subsection{Simulation B}

The second simulation is similar to the first, with the difference
being how the principal axis is generated. Our data matrix $X$ is
again simulated as rank one plus noise, and $X$, $E$, and $u$ follow
the relations on lines (\ref{Eq:X}), (\ref{Eq:E}), and
(\ref{Eq:u}). The difference is in how we create the principal
axis. For any branch $b$ in the phylogenetic tree, let
$\mathbf I_b \in \R^p$ be the indicator vector of the leaf nodes which
descend from $b$. Our principal axes $v$ are then defined as
\begin{align}
v &= \mathbf I_b / \sqrt{\mathbf I_b^T \mathbf 1}
\end{align}
We generated data matrices $X$ according to this scheme, varying both
$\sigma$ (the variance of the noise term) and $b$. We did one
simulation for each branch $b$ which has between 50 and 200 leaf nodes
as descendants. As before, we computed the correlation between the
true and estimated principal axis and the true and estimated sample
scores along the principal axis, and the results are shown in Figure
\ref{Fig:SimB}. In this simulation, we see that gPCA on $(X, Q, I)$
does the best when the number of leaf nodes associated with the
principal axis is high. Adaptive gPCA consistently outperforms both
gPCA on $(X, Q, I)$ and standard PCA in this setup, with the
performance not dropping off as quickly as standard PCA does in the
presence of increasing amounts of noise. 

\begin{figure}[h]
\begin{center}
\includegraphics[width=.85\textwidth]{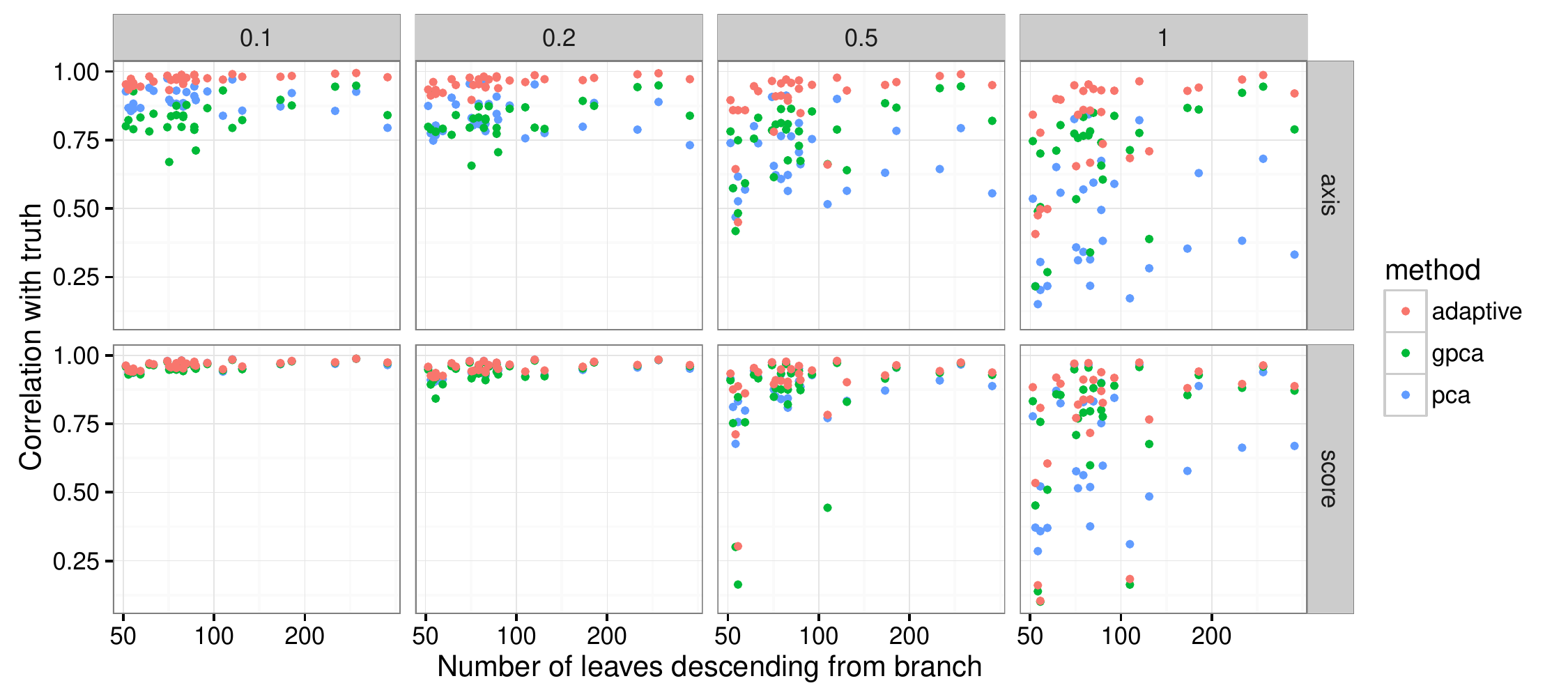}
\end{center}
\caption{Results from simulation B. Correlations between the true and
  estimated principal axis (top) and the true and estimated scores
  along the principal axis (bottom) for different levels of noise
  variance (columns labeled by noise variance). For each simulation,
  the principal axis is non-zero on all the leaves descending from a
  certain branch in the tree, and the $x$-axis gives the number of
  non-zero elements. }\label{Fig:SimB}
\end{figure}

In both of these simulations, the principal axes are structured
according to the tree in some sense, but in neither case is the data
generated according to the exact data model described in Section
\ref{Sec:agpcadatamodel}. This suggests that the method is not overly
dependent on the data coming from the exact model which was used to
motivate it and can perform well in a variety of situations.

\section{Real data example}\label{Sec:realdata}

To illustrate the method on real data, we return to the data set
described in Section \ref{Sec:motivation}. To review, the goal of the
study was to understand the effect of antibiotics on the gut
microbiome, and to this end fecal samples were taken from three
subjects over the course of several months, during which time each of
the subjects took two courses of the antibiotic
Ciprofloxacin. Bacterial abundances in the fecal samples were measured
using the procedure described in Section
\ref{Sec:motivation}. Measurements were made before the first course
of Cipro (called Pre Cp), during the first course of antibiotics (1st
Cp), in the week after the first course of antibiotics (1st WPC), more
than one week after the first course of antibiotics and before the
second course (Interim), during the second course of antibiotics (2nd
Cp), in the first week after the second course of antibiotics (2nd
WPC), and after that (Post Cp). For each of the samples we have the
abundances of approximately 2000 bacterial species and a tree
describing the phylogenetic relationships between the bacteria. We
looked at the results from adaptive gPCA, DPCoA, and standard PCA on
this data set. In adaptive gPCA, the similarity matrix $Q$ used to
incorporate the phylogeny is formed in the same way as for the
simulations (see equation (\ref{Eq:Qdef})) so that $Q_{ij}$ gives the
amount of shared ancestry between species $i$ and $j$.

\begin{figure}
\begin{center}
\includegraphics[width=\textwidth]{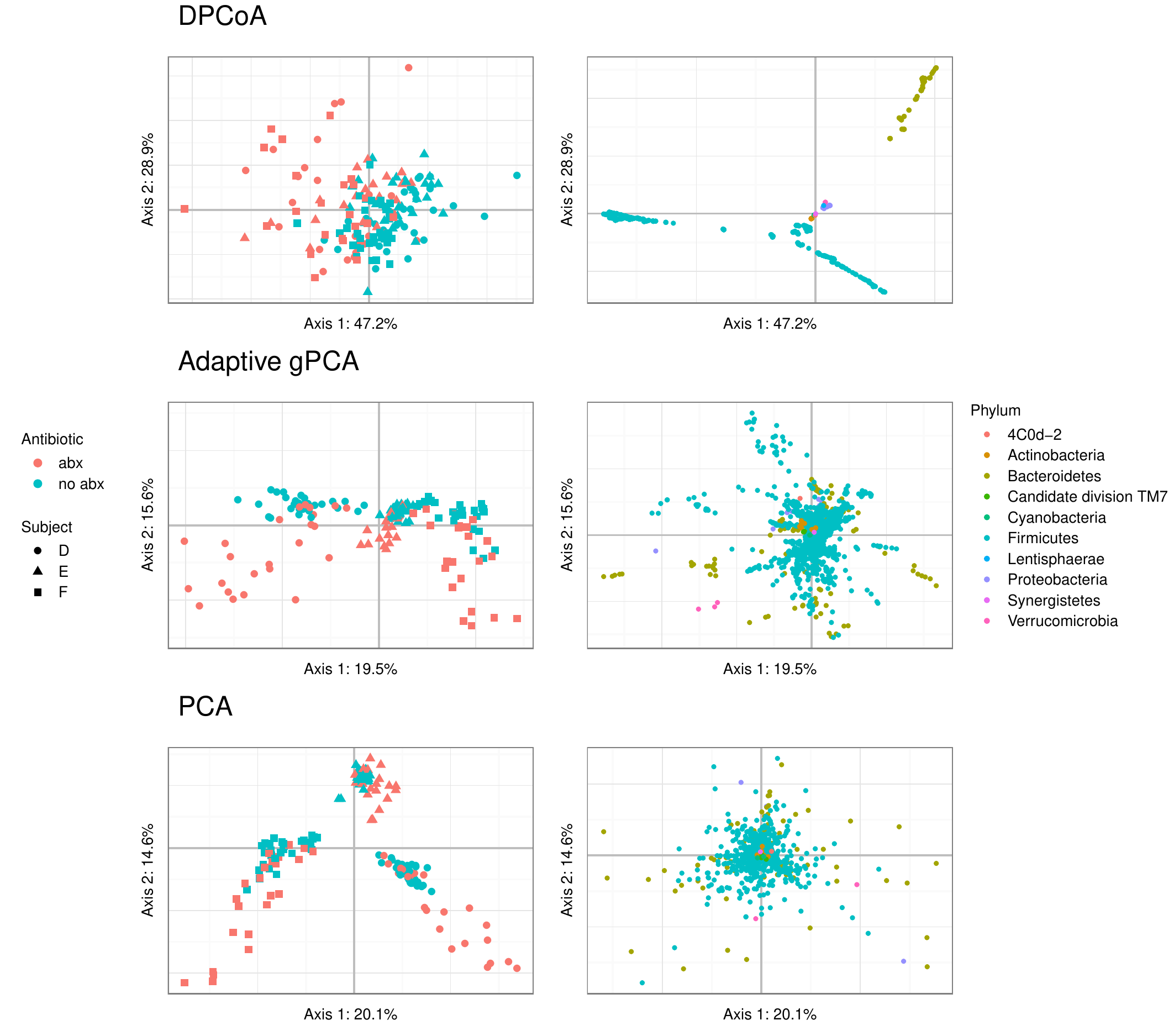}
\caption{Sample (left) and species (right) plots for DPCoA (top),
  adaptive gPCA (middle), and standard PCA (bottom). Colors in the
  sample plots represent a binning of the sample points into abx
  (either when the subject was on antibiotics or the week immediately
  after) or no abx (all other times). The colors in the species plots
  represent phyla. }\label{Fig:Antibiotic}
\end{center}
\end{figure}

Figure \ref{Fig:Antibiotic} shows the results of using the three
methods on this data set. The top pair of plots shows the results from
DPCoA, the middle from adaptive gPCA, and the bottom from standard
PCA. In each pair, the left-hand plot shows the sample scores on the
first and second principal axes, and the right-hand plot shows the
variable loadings on the first and second principal axes. All the
pairs of plots can be interpreted as biplots, so if a sample has a
large score on e.g. the first principal axis, we expect it to have
larger values for variables which have large loadings on the first
principal axis. 

The three methods give us quite different results. Just considering
the sample points to start with, in the DPCoA representation we see
some difference in the samples taken while the subjects were on
antibiotics compared with the others, but we see very little
difference between samples from the different subjects. PCA and
adaptive gPCA show complete separation between the samples from the
different subjects and a good degree of offset between the samples
taken while the subjects were on antibiotics compared with the
rest. It turns out that the second adaptive gPCA axis describes the
antibiotic perturbation very well: if we plot the scores along the
second axis over time, we see that the scores are stable when the
subjects are not on antibiotics, drop upon administration of the
antibiotic, and return to baseline when the antibiotic is stopped (see
Figure \ref{Fig:timecourse}).

\begin{figure}[t]
\begin{center}
\includegraphics[width=.8\textwidth]{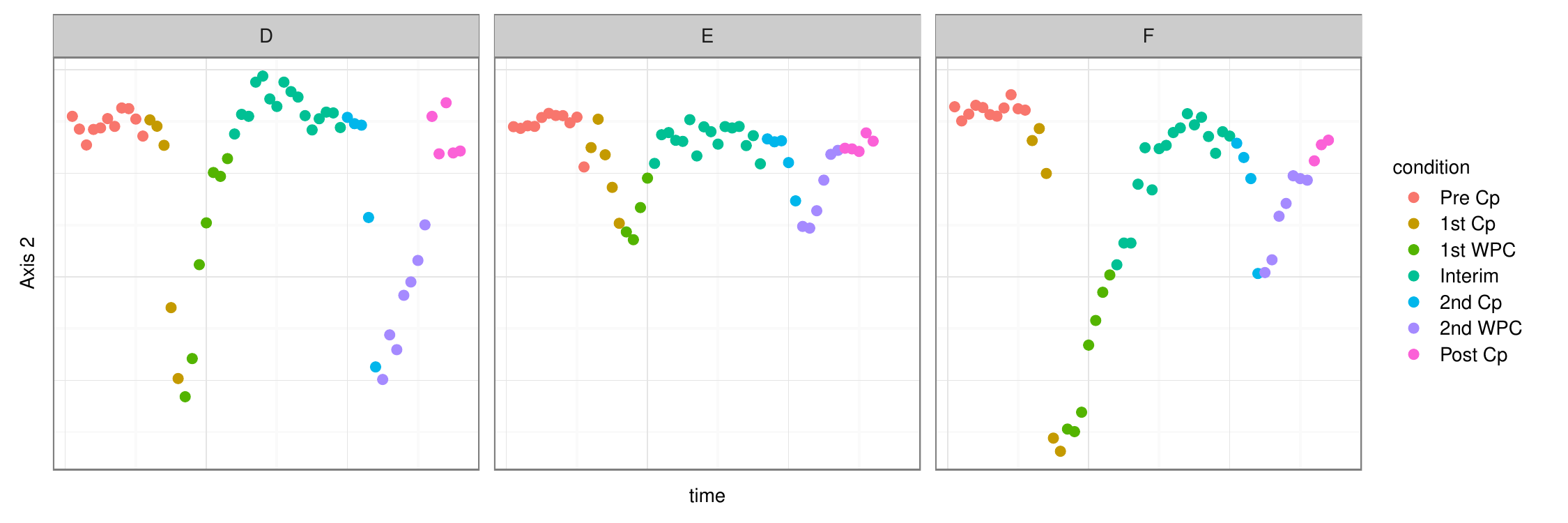}
\end{center}
\caption{A plot of the scores along the second axis from adaptive gPCA
  by time, plotted for each of the three individuals. We see very
  clearly that this axis is capturing species that change during the
  administration of the antibiotic but which are stable otherwise. The
  corresponding plots for PCA and DPCoA are much less
  compelling. }\label{Fig:timecourse}
\end{figure}

Turning next to the variable (species) loadings on the principal axes,
we see that the species points from PCA show no association with the
tree: species which are similar phylogenetically are no more likely to
have similar loadings on the principal axes than species which are
phylogenetically dissimilar. On the other end of the spectrum, the
species points from DPCoA are very related to the tree, and in
particular have loadings which are related to the deep branching
structure of the tree. We see this in the fact that the species points
from the two dominant phyla occupy completely disjoint areas in the
variable space. Adaptive gPCA gives results somewhere in the middle:
we see that species which are phylogenetically similar are more likely
to have similar loadings on the principal axes, but the phenomenon is
more local. Whereas in DPCoA, we have very large groups of similar
species with similar loadings on the principal axes (the two large
phyla), in adaptive gPCA we get smaller groups of similar species
having similar loadings on the principal axes.

Since the purpose of the study was to understand the effect of
antibiotics on the gut microbiome and since the second adaptive gPCA
axis seems to describe the disturbance due to the antibiotic, we can
look in more detail at the behavior of the species with large positive
or negative loadings on the second adaptive gPCA axis. The 27 species
with the largest positive scores along the second adaptive gPCA axis
are all of the genus {\em Faecalibacterium} (and in fact there are 28
members of this genus represented in the data set so this is nearly
the entire genus). Although different members of the genus are present
or absent in different subjects, when present they all show the same
pattern of declining in relative abundance during the treatment with
antibiotics and rebounding when the treatment is discontinued. This is
shown in the top row of Figure \ref{Fig:selectedotus}. Consistent with
what we see in Figure \ref{Fig:timecourse}, Subject E shows much less
of a disturbance compared to subjects D and F, and the disturbance in
F corresponding to the second course of antibiotics is much smaller
than that corresponding to the first.

Similarly, if we look at the 21 members of the Firmicutes phylum with
the largest negative scores along the second adaptive gPCA axis, we
see a similar phenomenon. Only 11 of the 21 species in this group are
classified at the genus level, but those 11 are all classified as {\em
  Blautia}, and all 21 species are classified at the family level as
Lachnospiraceae. These species tend to be even more subject specific
than those discussed above, with each species usually present in large
numbers in only one subject. However, when a species in this group is
present in a subject, its relative abundance tends to increase when
the antibiotic is administered and falls back to baseline when the
treatment is discontinued (shown in the bottom row of Figure
\ref{Fig:selectedotus}). This shows us another advantage of using a
method which incorporates phylogenetic information: Instead of having
a long list of species which may only be present in one subject and
whose behavior may not generalize to other individuals, we have a
clade whose members, when present, increase in relative abundance with
the administration of Cipro. This is a much more parsimonious
conclusion than that drawn from a list of unrelated taxa, and it is
straightforward to reason about and to test in later experiments.

\begin{figure}
\includegraphics[width=\textwidth]{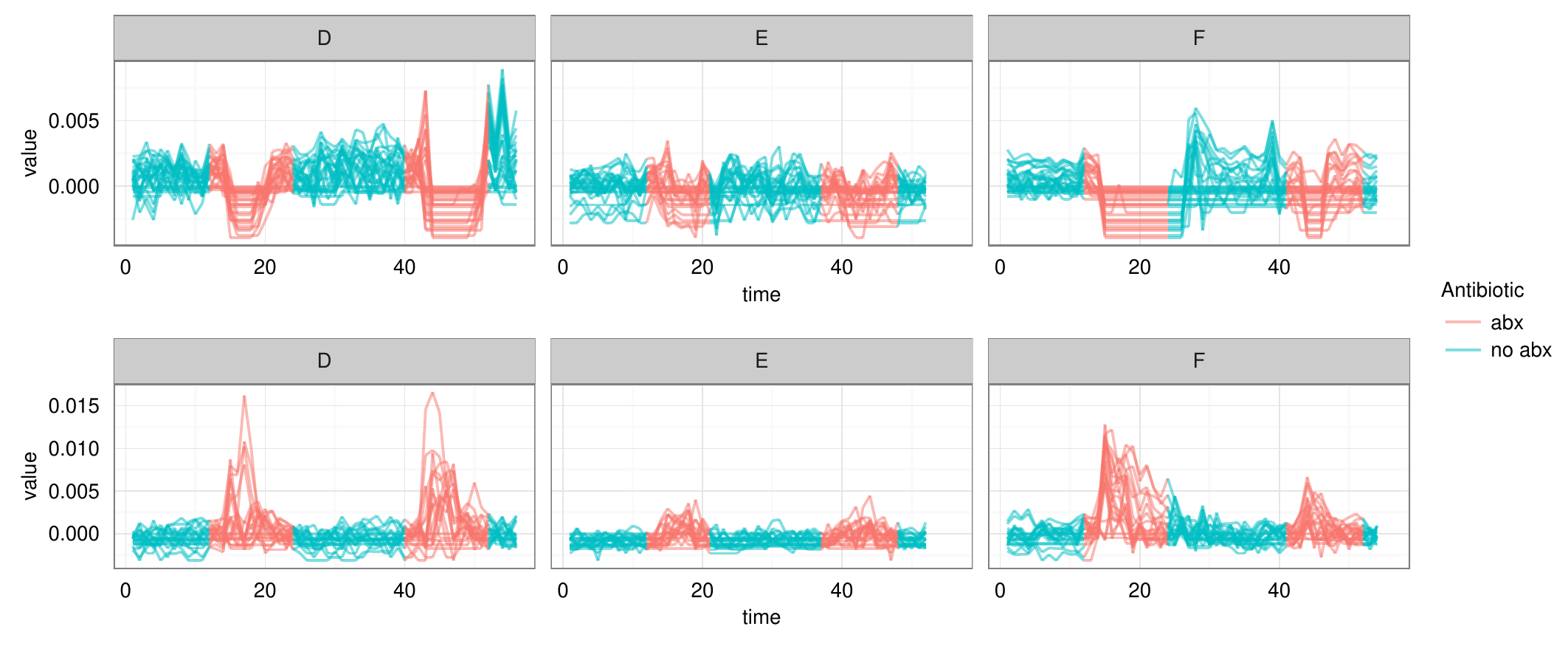}
\caption{Normalized abundances for two groups of species. Each line
  represents a species, each facet represents a subject. The top row
  shows the normalized abundances of each of 27 OTUs with the largest
  positive loadings on the second adaptive gPCA axis, and the bottom
  row shows the normalized abundances of the 21 Firmicutes with the
  largest negative loadings on the second adaptive gPCA axis. In
  general, the species with positive scores on the second axis see
  their relative abundances decline with the antibiotic treatment
  while the species with negative scores see their relative abundances
  increase. The sizes of the disturbances are consistent with what we
  see in Figure \ref{Fig:timecourse} (e.g. E has the smallest
  disturbance and the second antibiotic treatment for F leads to a
  smaller disturbance than the first). }\label{Fig:selectedotus}
\end{figure}

The results of this analysis show us some of the drawbacks of DPCoA
and standard PCA compared with adaptive gPCA. With standard PCA the
axes are difficult to interpret because of the lack of relationship
between the phylogenetic structure and the loadings of the variables
on the principal axes. DPCoA misses much of the true latent structure
in the data (it shows almost no subject effect and a smaller
antibiotic effect than either adaptive gPCA or standard PCA), which is
consistent with the simulations showing that DPCoA only performs well
in very limited situations. Adaptive gPCA recovers the latent
structure well and also has axes which are interpretable in terms of
small groups of related species. This sort of structure is useful to
scientists interested in understanding the underlying biology, and
looking in more detail at the groups of species associated with the
axes can give us insight into this biology and ideas about what steps
to take next.

\section{Conclusion}


In this paper, we have presented a method for creating low-dimensional
representations of a data matrix while taking into account side
information about the relationships between the variables. This is
done in a natural way by using a prior encoding the relationships
between the variables and performing PCA on the resulting posteriors,
taking into account the fact that the posteriors have non-spherical
variance. We show that performing PCA on the posterior estimates
obtained with this prior corresponds to a generalized PCA, with a
one-dimensional family of gPCAs arising from varying the prior
strength. A member of this family can then be picked by estimating the
scalings of the prior and the noise by maximum marginal likelihood. We
call the gPCA obtained in this manner adaptive gPCA.

Adaptive gPCA leads to a low-dimensional representation of the
samples. The loadings of similar variables along the principal axes in
this representation will be similar to each other, allowing the axes
to be more interpretable than in standard PCA. The effect is therefore
similar to what we would obtain by using PCA with a fused lasso
penalty on the variable loadings, but the motivation and derivation
are different, and our method is able to accommodate more general
variable structures than the fused lasso. Other attractive features of
our method are that we can obtain the global solution without worrying
about the algorithm being stuck in a local minimum and that we can
choose the amount of regularization to perform without having to
resort to potentially time-consuming cross-validation.

Using adaptive gPCA on a real data set shows us some of the advantages
of the method: we were able to identify the latent structure in the
data (the differences between the individuals and the antibiotic
treatment), and we were able to use the loadings of the variables on
the principal axes to understand the biology behind this latent
structure. For instance, the second adaptive gPCA axis was related to
the administration of the antibiotic, and examining the loadings of
the species along the second axis gave us groups of closely-related
species which share the same behavior upon administration of the
antibiotic. The implicit smoothing done by adaptive gPCA is helpful
here because not all the members of each group of species identified
by adaptive gPCA are present in each sample, but nonetheless the
members of the groups have similar behaviors when they are present.

It is also possible to extend adaptive gPCA in a number of
directions. If we have information about the precision with which
different variables or samples are measured, it is easy to incorporate
either sample of variable weights into the analysis. The family of
inner products described in this paper can also be used with other
methods which work in non-standard inner product spaces, such as
between- or within-class analysis (\cite{dray2015considering}), to
encourage loading vectors which are smooth according to the structure
of the variables. It can also be used in conjunction with formulations
for sparse gPCA (\cite{allen2014generalized}) to obtain
low-dimensional representations of the variables which are both sparse
and structured, and combining this further with between-class analysis
would yield a method for supervised learning with sparse and
structured variable loadings.

An R implementation of adaptive gPCA is available at
\begin{verbatim}
www.github.com/jfukuyama/adaptiveGPCA
\end{verbatim}
and can be installed in R with the command
\begin{verbatim}
devtools::install_github("jfukuyama/adaptiveGPCA")
\end{verbatim}
The package allows for either the automatic selection procedure
described in Section \ref{Sec:adaptive} or for manual
selection. Manual selection is mediated by a shiny gadget
\cite{shiny}, which provides an interactive plot with a slider bar
allowing the user to move easily between visualizations corresponding
to different prior strengths. The package also includes the antibiotic
data used in this paper and a vignette which reproduces the analysis.

\appendix

\section{Proof of Theorem 1}\label{app1}

The posterior distribution of all of the
$\mathbf \mu_i$'s given the data follows a matrix normal distribution
$\mathcal{MN}_{n \times p} (\sigma_2^{-2}XS, I_n, S)$. Therefore,
following the structured error interpretation of gPCA, to take into
account the error structure we should perform gPCA on the triple
$(\sigma_2^{-2} XS, S^{-1}, I_n)$. Since we are interested in the
low-dimensional representation of the samples and variables, the
scaling is not important and going forward we will drop the
$\sigma_2^{-2}$ factor and consider gPCA on $(XS, S^{-1}, I_n)$. 

Now, note that the sample scores obtained by gPCA on
$(XS, S^{-1}, I_n)$ are the same as those obtained by gPCA on
$(X, S, I_n)$, as can be verified by plugging both sets of variables
into the optimization problem in (\ref{Eq:gpca_scores}). The principal
axes from $(XS, S^{-1}, I_n)$ are equal to the principal axes from
$(X, S, I_n)$ transformed by $S$. To see the equivalence, note that
for the principal axes from the triple $(XS, S^{-1}, I)$, we need to
solve the problem
\begin{align}
\max_{\tilde v_i \in \R^p} & \quad \tilde v_i^T X^T X \tilde v_i ,
                             \quad i = 1,\ldots, k \label{Eq:bayes_pca_1}\\
\text{s.t. }& \quad \tilde v_i^T S^{-1} \tilde v_i = 1, \quad i = 1,\ldots, k\nonumber \\
& \quad \tilde v_i^T S^{-1} \tilde v_j = 0, \quad 1 \le i < j \le k\nonumber
\end{align}
For the principal axes on the triple $(X, S, I)$, we need to solve
\begin{align}
\max_{v_i \in \R^p} & \quad v_i^T S X^T X S v_i , \quad i = 1,\ldots,
                      k \label{Eq:bayes_pca_2}\\
\text{s.t. }& \quad v_i^T S v_i = 1, \quad i = 1,\ldots, k \nonumber \\
& \quad  v_i^T S v_j = 0, \quad 1 \le i < j \le k \nonumber
\end{align}
Now if we make the change of variables $\tilde v_i = S v_i$, problems
(\ref{Eq:bayes_pca_1}) and (\ref{Eq:bayes_pca_2}) are the same.

\section{Proof of Theorem 2}\label{app}

The proof here follows almost exactly from \cite{purdom2011analysis},
but I am including a full proof for completeness.

First some notation. For a weight vector $w$ satisfying
$w^T \mathbf 1 = 1$, let $P_w = I - \mathbf 1 w^T$ represent the
weighted centering operator, and let $D_w$ be the diagonal matrix with
$w$ along the diagonal. We have $p$ variables measured on $n$
samples. Let $C \in \R^{n \times p}$ be our original data matrix with
$C_{ij}$ containing the count of variable $j$ for sample $i$, and let
$w_L$ and $w_S$ denote sample and variable weights,
respectively. These are obtained by normalizing the row sums and
column sums, so $w_L = C \mathbf 1 / \mathbf 1^T C \mathbf 1$ and
$w_S = C^T \mathbf 1 / \mathbf 1^T C \mathbf 1$. Let
$X \in \R^{n \times p}$ be the matrix with frequency profiles for each
sample, so $X = D_{w_L}^{-1}C$. Finally, let the matrix
$\delta \in \R^{p \times p}$ contain the squared Euclidean distances
between variables. We are assuming that these distances are Euclidean.

\subsubsection*{DPCoA}
For the first step of DPCoA, we get the variable locations from
classical multi-dimensional scaling. The weighted version of
multi-dimensional scaling is obtained by finding the
eigendecomposition of
$D_{w_S}^{1/2} P_{w_S} (-\delta / 2) P_{w_S}^T D_{w_S}^{1/2}$. Then we
have
\begin{align}
U \Lambda U^T &= D_{w_S}^{1/2} P_{w_S} (-\delta /2) P_{w_S}^T
                D_{w_S}^{1/2}\\
Z &= D_{w_S}^{-1/2} U \Lambda^{1/2} \\
Y &= XZ
\end{align}
$Z$ is then a matrix in $\R^{p \times d}$ ($d$ the dimension of the
space the points are embedded in, $d < p$) containing the coordinates
of the variable points given by multi-dimensional scaling. Since the
rows of $X$ contain the frequencies of the variables at each location,
the rows of $Y$ contain the barycenters of the variable clouds
corresponding to each sample.

The second step of DPCoA, now that we have the barycenters of each
sample in $Y$, is to do PCA on the triple $(Y, I, D_{w_L})$. This
means we have to solve
\begin{align}
Y^T D_{w_L} Y M &= M \Lambda & M^T M &= I\\
Y Y^T D_{w_L} L &= L \Lambda & L^T D_{w_L} L &= I \label{conditionsL}
\end{align}
The sample scores are then found in $L \Lambda^{1/2}$ and the variable
scores are found in $ZM$. We can rewrite the first equation in
(\ref{conditionsL}) as
\begin{align}
L \Lambda &= YY^T D_{w_L} L\\
 &= XZZ^T X^T D_{w_L}L \\
&= X D^{-1/2}_{w_S} U \Lambda U^T D_{w_S}^{-1/2} X^T D_{w_L}L \\
&= X P_{w_S} (-\delta / 2) P_{w_S}^T X^T D_{w_L} L
\end{align}
Then since $P_{w_S}$ is a projection operator, $P_{w_S} =
P_{w_S}P_{w_S}$ and so the previous line can be rewritten as
\begin{align}
L \Lambda = X P_{w_S} (-P_{w_S} \delta P_{w_S}^T / 2)
  P_{w_S}^TX^T D_{w_L } L\label{conditionsLrewritten}
\end{align}

\subsubsection*{Generalized PCA}

First of all, recall that since $X$ is a contingency table, centering
$X$ by row and centering $X$ by column are the same, $XP_{w_S}=
P_{w_L}^T X$. Call this centered matrix $\tilde X$. 

Now consider generalized PCA of the triple $(\tilde X, Q, D_{w_L})$
where $\tilde X$ is a column-centered version of $X$, so $\tilde X = X
P_{w_S}$ and $Q = P_{w_S}(-\delta / 2) P_{w_S}^T$. The equations that
need to be satisfied for this gPCA are 
\begin{align}
\tilde X^T D_{w_L} \tilde X Q A &= A \Psi & A^T Q A &= I\\
\tilde X Q \tilde X^T D_{w_L} B &= B \Psi  & B^T D_{w_L}B &= I \label{conditionsB}
\end{align}
$B \Psi^{1/2}$ contains the sample scores from the gPCA. By comparing
line (\ref{conditionsB}) and line (\ref{conditionsLrewritten}), we see
that the conditions for the pair $B, \Psi$ and the pair $L, \Lambda$
are the same, and so the sample scores from gPCA and DPCoA are the
same up to a sign change.

Then the variable scores given by DPCoA are given by $ZM$. The
generalized SVD tells us that $Y = L \Lambda^{1/2} M^T$, which, along
with $M^T M = I$ and $L^T D_{w_L}L = I$ implies that
$M^T = \Lambda^{-1/2} L^T D_{w_L}Y$. The generalized SVD of $\tilde X$
is $\tilde X = B \Psi^{1/2} A^T$, which, along with the corresponding
orthogonality conditions, implies that
$\tilde X Q A \Psi^{-1} = B \Psi^{-1/2}$. Then we can rewrite the
variable scores $ZM$ as
\begin{align}
ZM &= Z Y^T D_{w_L} L \Lambda^{-1/2}\\
&= ZZ^T X^T D_{w_L} L \Lambda^{-1/2} \\
&= D_{w_S}^{-1/2} U \Lambda U^T D_{w_S}^{-1/2} X^T D_{w_L} L
  \Lambda^{-1/2}\\
&= P_{w_S}(-\delta / 2) P_{w_S}^T X^T D_{w_L} B \Psi^{-1/2} \\
&= Q \tilde X^T D_{w_L} B \Psi^{-1/2} \\
&= Q \tilde X^T D_{w_L}\tilde X Q A \Psi^{-1} \\
&= QA
\end{align}
So we can get the variable scores from DPCoA by multiplying the
variable scores from gPCA by $Q$.

\bibliographystyle{natbib}
\bibliography{adaptivegpca}

\begin{thebibliography}{}

\bibitem[Allen {\em et~al.}(2014)Allen, Grosenick, and
  Taylor]{allen2014generalized}
Allen, G.~I., Grosenick, L., and Taylor, J. (2014).
\newblock A generalized least-square matrix decomposition.
\newblock {\em Journal of the American Statistical Association\/}, {\bf
  109}(505), 145--159.

\bibitem[Brenner {\em et~al.}(2005)Brenner, Staley, and
  Krieg]{brenner2005classification}
Brenner, D.~J., Staley, J.~T., and Krieg, N.~R. (2005).
\newblock Classification of procaryotic organisms and the concept of bacterial
  speciation.
\newblock In {\em Bergey’s Manual of Systematic Bacteriology\/}, pages
  27--32. Springer.

\bibitem[Callahan {\em et~al.}(2016)Callahan, Sankaran, Fukuyama, McMurdie, and
  Holmes]{callahan2016bioconductor}
Callahan, B.~J., Sankaran, K., Fukuyama, J.~A., McMurdie, P.~J., and Holmes,
  S.~P. (2016).
\newblock Bioconductor workflow for microbiome data analysis: from raw reads to
  community analyses.
\newblock {\em F1000Research\/}, {\bf 5}.

\bibitem[Chang {\em et~al.}(2011)Chang, Luan, and Sun]{chang2011variance}
Chang, Q., Luan, Y., and Sun, F. (2011).
\newblock Variance adjusted weighted unifrac: a powerful beta diversity measure
  for comparing communities based on phylogeny.
\newblock {\em BMC Bioinformatics\/}, {\bf 12}(1), 1.

\bibitem[Chang {\em et~al.}(2016)Chang, Cheng, Allaire, Xie, and
  McPherson]{shiny}
Chang, W., Cheng, J., Allaire, J., Xie, Y., and McPherson, J. (2016).
\newblock {\em shiny: Web Application Framework for R\/}.
\newblock R package version 0.13.2.

\bibitem[Chen {\em et~al.}(2012)Chen, Bittinger, Charlson, Hoffmann, Lewis, Wu,
  Collman, Bushman, and Li]{chen2012associating}
Chen, J., Bittinger, K., Charlson, E.~S., Hoffmann, C., Lewis, J., Wu, G.~D.,
  Collman, R.~G., Bushman, F.~D., and Li, H. (2012).
\newblock Associating microbiome composition with environmental covariates
  using generalized unifrac distances.
\newblock {\em Bioinformatics\/}, {\bf 28}(16), 2106--2113.

\bibitem[Cohan(2002)Cohan]{cohan2002bacterial}
Cohan, F.~M. (2002).
\newblock What are bacterial species?
\newblock {\em Annual Reviews in Microbiology\/}, {\bf 56}(1), 457--487.

\bibitem[Dethlefsen and Relman(2011)Dethlefsen and
  Relman]{dethlefsen2011incomplete}
Dethlefsen, L. and Relman, D.~A. (2011).
\newblock Incomplete recovery and individualized responses of the human distal
  gut microbiota to repeated antibiotic perturbation.
\newblock {\em Proceedings of the National Academy of Sciences\/}, {\bf
  108}(Supplement 1), 4554--4561.

\bibitem[Doolittle and Papke(2006)Doolittle and Papke]{doolittle2006genomics}
Doolittle, W.~F. and Papke, R.~T. (2006).
\newblock Genomics and the bacterial species problem.
\newblock {\em Genome Biology\/}, {\bf 7}(9), 1.

\bibitem[Dray {\em et~al.}(2015)Dray, Pavoine, and Aguirre~de
  C{\'a}rcer]{dray2015considering}
Dray, S., Pavoine, S., and Aguirre~de C{\'a}rcer, D. (2015).
\newblock Considering external information to improve the phylogenetic
  comparison of microbial communities: a new approach based on constrained
  double principal coordinates analysis (cdpcoa).
\newblock {\em Molecular Ecology Resources\/}, {\bf 15}(2), 242--249.

\bibitem[Edgar(2010)Edgar]{edgar2010search}
Edgar, R.~C. (2010).
\newblock Search and clustering orders of magnitude faster than blast.
\newblock {\em Bioinformatics\/}, {\bf 26}(19), 2460--2461.

\bibitem[Holmes(2008)Holmes]{holmes2008multivariate}
Holmes, S. (2008).
\newblock Multivariate data analysis: the \uppercase{F}rench way.
\newblock {\em Probability and Statistics: Essays in Honor of David A.
  Freedman. Institute of Mathematical Statistics, Beachwood, Ohio\/}, pages
  219--233.

\bibitem[Johnstone and Lu(2012)Johnstone and Lu]{johnstone2012consistency}
Johnstone, I.~M. and Lu, A.~Y. (2012).
\newblock On consistency and sparsity for principal components analysis in high
  dimensions.
\newblock {\em Journal of the American Statistical Association\/}.

\bibitem[Kondor and Lafferty(2002)Kondor and Lafferty]{kondor2002diffusion}
Kondor, R.~I. and Lafferty, J. (2002).
\newblock Diffusion kernels on graphs and other discrete structures.
\newblock In {\em Proceedings of the 19th International Conference on Machine
  Learning\/}, pages 315--322.

\bibitem[Li and Li(2008)Li and Li]{li2008network}
Li, C. and Li, H. (2008).
\newblock Network-constrained regularization and variable selection for
  analysis of genomic data.
\newblock {\em Bioinformatics\/}, {\bf 24}(9), 1175--1182.

\bibitem[Lozupone and Knight(2005)Lozupone and Knight]{lozupone2005unifrac}
Lozupone, C. and Knight, R. (2005).
\newblock Unifrac: a new phylogenetic method for comparing microbial
  communities.
\newblock {\em Applied and Environmental Microbiology\/}, {\bf 71}(12),
  8228--8235.

\bibitem[Lozupone {\em et~al.}(2007)Lozupone, Hamady, Kelley, and
  Knight]{lozupone2007quantitative}
Lozupone, C.~A., Hamady, M., Kelley, S.~T., and Knight, R. (2007).
\newblock Quantitative and qualitative $\beta$ diversity measures lead to
  different insights into factors that structure microbial communities.
\newblock {\em Applied and Environmental Microbiology\/}, {\bf 73}(5),
  1576--1585.

\bibitem[Matsen and Evans(2013)Matsen and Evans]{matsen2013edge}
Matsen, F.~A. and Evans, S.~N. (2013).
\newblock Edge principal components and squash clustering: using the special
  structure of phylogenetic placement data for sample comparison.
\newblock {\em PLoS ONE\/}.

\bibitem[McMurdie and Holmes(2014)McMurdie and Holmes]{mcmurdie2014waste}
McMurdie, P.~J. and Holmes, S. (2014).
\newblock Waste not, want not: why rarefying microbiome data is inadmissible.
\newblock {\em PLoS Comput Biol\/}, {\bf 10}(4), e1003531.

\bibitem[Paradis {\em et~al.}(2004)Paradis, Claude, and
  Strimmer]{paradis2004ape}
Paradis, E., Claude, J., and Strimmer, K. (2004).
\newblock A{PE}: analyses of phylogenetics and evolution in {R} language.
\newblock {\em Bioinformatics\/}, {\bf 20}, 289--290.

\bibitem[Pavoine {\em et~al.}(2004)Pavoine, Dufour, and
  Chessel]{pavoine2004dissimilarities}
Pavoine, S., Dufour, A.-B., and Chessel, D. (2004).
\newblock From dissimilarities among species to dissimilarities among
  communities: a double principal coordinate analysis.
\newblock {\em Journal of Theoretical Biology\/}, {\bf 228}(4), 523--537.

\bibitem[Purdom(2011)Purdom]{purdom2011analysis}
Purdom, E. (2011).
\newblock Analysis of a data matrix and a graph: Metagenomic data and the
  phylogenetic tree.
\newblock {\em The Annals of Applied Statistics\/}, {\bf 5}(4), 2326--2358.

\bibitem[Quast {\em et~al.}(2013)Quast, Pruesse, Yilmaz, Gerken, Schweer,
  Yarza, Peplies, and Gl{\"o}ckner]{quast2013silva}
Quast, C., Pruesse, E., Yilmaz, P., Gerken, J., Schweer, T., Yarza, P.,
  Peplies, J., and Gl{\"o}ckner, F.~O. (2013).
\newblock The silva ribosomal rna gene database project: Improved data
  processing and web-based tools.
\newblock {\em Nucleic Acids Research\/}, {\bf 41}(D1), D590--D596.

\bibitem[Randolph {\em et~al.}(2015)Randolph, Zhao, Copeland, Hullar, and
  Shojaie]{randolph2015kernel}
Randolph, T.~W., Zhao, S., Copeland, W., Hullar, M., and Shojaie, A. (2015).
\newblock Kernel-penalized regression for analysis of microbiome data.
\newblock {\em arXiv preprint arXiv:1511.00297\/}.

\bibitem[Rapaport {\em et~al.}(2007)Rapaport, Zinovyev, Dutreix, Barillot, and
  Vert]{rapaport2007classification}
Rapaport, F., Zinovyev, A., Dutreix, M., Barillot, E., and Vert, J.-P. (2007).
\newblock Classification of microarray data using gene networks.
\newblock {\em BMC Bioinformatics\/}, {\bf 8}(1), 35.

\bibitem[Rinaldo {\em et~al.}(2009)Rinaldo {\em et~al.}]{rinaldo2009properties}
Rinaldo, A. {\em et~al.} (2009).
\newblock Properties and refinements of the fused lasso.
\newblock {\em The Annals of Statistics\/}, {\bf 37}(5B), 2922--2952.

\bibitem[Shade and Handelsman(2012)Shade and Handelsman]{shade2012beyond}
Shade, A. and Handelsman, J. (2012).
\newblock Beyond the venn diagram: the hunt for a core microbiome.
\newblock {\em Environmental Microbiology\/}, {\bf 14}(1), 4--12.

\bibitem[Silverman(1996)Silverman]{silverman1996smoothed}
Silverman, B.~W. (1996).
\newblock Smoothed functional principal components analysis by choice of norm.
\newblock {\em The Annals of Statistics\/}, {\bf 24}(1), 1--24.

\bibitem[Tibshirani and Wang(2008)Tibshirani and Wang]{tibshirani2008spatial}
Tibshirani, R. and Wang, P. (2008).
\newblock Spatial smoothing and hot spot detection for cgh data using the fused
  lasso.
\newblock {\em Biostatistics\/}, {\bf 9}(1), 18--29.

\bibitem[Tibshirani {\em et~al.}(2005)Tibshirani, Saunders, Rosset, Zhu, and
  Knight]{tibshirani2005sparsity}
Tibshirani, R., Saunders, M., Rosset, S., Zhu, J., and Knight, K. (2005).
\newblock Sparsity and smoothness via the fused lasso.
\newblock {\em Journal of the Royal Statistical Society: Series B (Statistical
  Methodology)\/}, {\bf 67}(1), 91--108.

\bibitem[Witten {\em et~al.}(2009)Witten, Tibshirani, and
  Hastie]{witten2009penalized}
Witten, D.~M., Tibshirani, R., and Hastie, T. (2009).
\newblock A penalized matrix decomposition, with applications to sparse
  principal components and canonical correlation analysis.
\newblock {\em Biostatistics\/}, {\bf 10}(3), 515--534.

\end{thebibliography}

\end{document}